\documentclass[runningheads]{llncs}
\usepackage[T1]{fontenc}
\usepackage{graphicx}
\usepackage{cite}
\usepackage{amsmath,amssymb,amsfonts}
\usepackage{textcomp}
\usepackage{xcolor}
\usepackage[hidelinks]{hyperref}
\usepackage{booktabs}
\usepackage{array}
\usepackage{microtype}
\usepackage{subcaption}
\usepackage{url}
\usepackage{xurl}

\begin{document}
\title{Redesigning Trust: Replacing Dark Patterns with Fair Choice Architecture in Financial Interfaces}
\titlerunning{Replacing Dark Patterns with Fair Choice Architecture}
%
\author{Oluwadamilola Awakan\orcidID{0009-0007-4950-4073} \and Tawan Aroonwechkul\orcidID{0009-0008-7561-1910} \and
Roshan Gunjoor\orcidID{0009-0005-2712-0250} \and
Supriya Khadka\orcidID{0009-0008-4126-624X} \and
Sanchari Das\orcidID{0000-0003-1299-7867}}

%
\authorrunning{Awakan et al.}
%
\institute{George Mason University\\
\email{\{oawakan,taroonwe,rgunjoor,skhadk,sdas35\}@gmu.edu}}
\maketitle              
\begin{abstract}
Digital financial platforms make enrollment effortless and cancellation laborious. This asymmetry is a dark pattern that manipulates users who have already decided to leave. Existing work identifies such patterns after deployment, and regulators sanction them after harm, yet neither provides designers with a criterion for building interfaces that avoid manipulation. We model the provider as an adversary whose instrument is effort and express fairness as a constraint requiring that leaving never cost more than joining. Defining interaction cost over navigation steps, mandatory inputs, and confirmation prompts, we prove that this constraint holds for every assignment of effort weights if and only if no component of the exit flow exceeds its counterpart at entry. Fairness is therefore verifiable by counting rather than by estimating cognitive effort, and exact equivalence is unnecessary because exit legitimately requires fewer inputs than entry. We instantiate the model, together with invariants for visual parity and linguistic neutrality, in a mobile credit card prototype with parallel sign-up and cancellation workflows.

\keywords{Dark Patterns \and Choice Architecture \and Interaction Cost \and User Autonomy \and Human-Centered Security}

\end{abstract}
\section{Introduction}
The design of digital financial services shapes consumer behavior as decisively as the terms of the services themselves~\cite{shahid2022examining, maziriri2025customer, kishnani2023assessing}. Credit card platforms make this visible through a persistent asymmetry. Enrollment is engineered to be frictionless and requires only a few steps, whereas cancellation is obstructed with additional steps, ambiguous options, and obscured pathways~\cite{sheil2024staying, nembaware2025dark}. Because the effort required to leave exceeds the effort required to join, users may remain enrolled longer than they intend, exposing them to ongoing financial obligations that conflict with their original choices~\cite{sa2025hidden}.

This asymmetry reflects a dark pattern, a design strategy that steers users toward decisions contrary to their intended goals~\cite{gray2018dark, mathur2019dark}. Unlike conventional usability shortcomings, the resulting friction is intentionally introduced to increase retention~\cite{runge2023dark}. By exploiting visual hierarchy, interface wording, and users' tendency to avoid effortful tasks, such designs can discourage cancellation even when users intend to proceed. In the context of credit card services, where continued enrollment may result in recurring fees, interest charges, and other financial obligations, these effects have significant financial consequences~\cite{brenncke2024regulating, dinner2011partitioning}. Because these interfaces interfere with users' ability to make informed and intentional decisions about their financial accounts, protecting against such manipulation is a human-centered security challenge as much as a usability one. Addressing this challenge requires more than detecting manipulative interfaces after deployment; it requires a principled way to prevent them during design.

To this end, we propose a symmetric interaction model that structurally aligns the onboarding and offboarding of a financial service, ensuring that leaving is never more difficult than joining. We define symmetry in terms of interface dimensions that providers directly control and can manipulate, including the number of navigation steps and the number of confirmation prompts. Fairness is expressed as a single ordinal constraint in which the interaction cost of offboarding must not exceed that of onboarding. This formalization recasts fairness from an ethical aspiration into a measurable property that can be specified as a design invariant and audited against an implemented interface. We instantiate the model in a high-fidelity mobile prototype with parallel sign-up and cancellation workflows and verify it against the constraint by structured inspection.

Specifically, this paper makes the following contributions:
\begin{itemize}
    \item We formulate fair choice architecture as an enforceable design requirement, where the cost of leaving a service must never exceed the cost of joining it on the dimensions a provider controls.
    \item We introduce a measurable interaction-cost formulation over navigation steps, mandatory inputs, and confirmation prompts, and show that the verdict of the fairness constraint is independent of the effort weights, reducing verification to componentwise counting.
    \item We develop a mobile credit card prototype demonstrating how structural, visual, and linguistic symmetry can be specified as design invariants and audited in an implemented interface.
\end{itemize}

\section{Related Work}
Dark patterns have been documented extensively across digital interfaces, with taxonomies characterizing their mechanics and empirical work measuring their prevalence and effects~\cite{gray2018dark, luguri2021shining, lacey2023clusters, zac2025dark, kumar2025dark}. Within these taxonomies, asymmetry is a defining attribute of manipulative design, describing interfaces that make the provider-favored option easier to select than the alternative~\cite{mathur2019dark, mathur2021makes}. Recent work consolidates these overlapping vocabularies into a shared ontology spanning academic and regulatory sources and distinguishes low-level interface patterns, which afford detection, from higher-level strategies~\cite{gray2024ontology}. Cancellation friction sits at this low level, and behavioral economics names it directly. Sludge denotes friction in choice architecture that makes it harder for people to obtain beneficial outcomes, with subscription cancellations buried in menus as a canonical example~\cite{thaler2018nudge, sunstein2022sludge}. 

These concerns have also prompted regulatory action. The U.S. Federal Trade Commission has pursued enforcement against manipulative subscription-cancellation practices, including its 2023 complaint against Amazon over a flow requiring more steps to cancel than to enroll~\cite{ftc2023amazon}. Its 2024 Negative Option Rule required cancellation to be as simple as enrollment~\cite{ftc2024negativeoption}. The Eighth Circuit vacated the rule in 2025 on procedural grounds because the Commission had not conducted a required preliminary cost-benefit analysis, without reaching substantive challenges~\cite{customcomms2025}. The Commission has since initiated new rulemaking on negative-option plans~\cite{ftc2026anprm}, while state automatic-renewal statutes remain in force and increasingly prescribe permitted retention steps.

The principle is therefore named but unmeasured. Neither the regulatory text nor the surrounding literature specifies what \textit{as simple as} means operationally, and so compliance rests on case-by-case judgment about when friction becomes excessive. The same gap appears in research. Prior treatments of asymmetry are descriptive, employing it to classify interfaces that already exist and to detect patterns after deployment~\cite{mathur2019dark, runge2023dark}, with related step-based audits applied to account remediation workflows~\cite{markert2023transcontinental}. Behavioral work calls for sludge audits that measure the time and hassle a process imposes, yet supplies no threshold that separates an acceptable exit from a manipulative one~\cite{sunstein2022sludge}. Detection establishes whether manipulation is present, but it does not provide designers with a constructive criterion for building interfaces that avoid manipulation. Work on financial interfaces likewise documents cancellation friction without prescribing the structure that would remove it~\cite{sheil2024staying, nembaware2025dark}.

We supply the missing operationalization. Prior empirical work on cancellation friction, including click-count comparisons between subscription and cancellation flows~\cite{sheil2024staying}, documents asymmetry post hoc in deployed interfaces and cautions that click counts alone may not capture experienced difficulty. Our contribution differs in kind rather than degree: we do not propose an improved effort metric, but a design-time constraint whose compliance verdict is provably invariant to the choice of effort weights (Proposition~\ref{prop:wi}), so that the very ambiguity in weighting that limits click-count measures does not affect whether a workflow passes or fails. Rather than treating asymmetry as an attribute to be recognized, we treat symmetry as a constraint to be enforced, expressed as an ordering over interaction cost that a designer can build toward and an auditor can verify.

\section{The Symmetric Interaction Model}

\subsection{Threat Model}
We consider the service provider as the adversary. Unlike an external attacker, the provider controls the interface itself and therefore requires no vulnerability, no compromised credential, and no deception about matters of fact. The provider's objective is to retain a user whose intention is to leave, and the attack surface is the user's decision process rather than the system's technical boundary.

The provider exercises its control through the structure of the exit workflow. It may lengthen the path to cancellation, insert dialogues that invite reconsideration, de-emphasize the control that completes the exit, and phrase confirmations so that continuing appears prudent and leaving appears rash. Each addition is defensible in isolation, and their accumulation is the attack. The user retains formal freedom to cancel throughout; what the provider removes is the practical exercise of that freedom.

We assume the user has already formed the intention to cancel, since the design of interfaces that shape whether that intention forms is a separate problem. We further assume the provider completes cancellation once the user reaches the end of the flow, since outright refusal is a distinct and more legible violation than the effort-based manipulation we model. The manipulation is thus confined to the effort imposed along the way, which is precisely why a defense must be stated in terms of effort.

\subsection{Quantifying Interaction Cost}
\label{sec:quantifying}
A workflow is characterized by three interaction quantities. \textbf{Steps} ($S$) are structural navigation actions such as advancing a page, opening a menu, or selecting a standard control. Concretely, one step corresponds to one transition from a given screen to the next screen in the flow; a screen requiring several simultaneous actions, such as completing a multi-field form, still counts as a single step, since the number of fields is captured separately by $I$. Optional screens that a user may bypass without completing the flow's defining action are excluded from $S$, consistent with the treatment of optional inputs described below. \textbf{Inputs} ($I$) are mandatory data-entry fields requiring the user to supply information. \textbf{Prompts} ($P$) are cognitive interruptions, including confirmation dialogues, retention offers, and secondary warnings. The interaction cost of a workflow is the weighted sum

\[
C = w_s S + w_i I + w_p P, \qquad w_s, w_i, w_p \geq 0,
\]

where the weights express the relative cognitive and physical effort of each action type. We write $C_{on}$ and $C_{off}$ for the cost of onboarding and offboarding.

Both $S_{on}$ and $S_{off}$ are measured from the authenticated home screen to the first screen where the defining action, enrollment submission or cancellation request, is available. Steps needed to locate a flow’s entry point count toward that flow’s cost. This prevents providers from satisfying the componentwise constraint on a narrow subflow while burying its entry point behind undisclosed navigation.

A dialog is counted as a single prompt only if it presents one binary decision, proceed or abandon, with no more than a fixed word budget of accompanying text and no embedded data entry. A dialog that requires document upload, multi-field input, an external wait, or a channel switch (e.g., a phone call) is coded under its native category, input, step, or a distinct step for the wait or switch, rather than folded into a single prompt unit. This atomicity rule prevents a provider from consolidating multiple burdens into one nominally-counted prompt.

These quantities are not equally available to the adversary. A provider may insert arbitrarily many navigation steps and retention prompts into an exit flow without collecting a single new item of information, so $S$ and $P$ are manipulable at will. Mandatory inputs behave differently, since enrollment legitimately requires identity and financial data that the system already possesses at cancellation. Inflating $I$ during exit would mean demanding information the provider already holds, which is conspicuous and independently actionable. Symmetry must therefore be defined over the dimensions the adversary can inflate without detection, rather than over the scalar total.

\begin{definition}[Structural Symmetry]
    A pair of workflows is structurally symmetric when $S_{off} \leq S_{on}$, $P_{off} \leq P_{on}$, and additionally $S_{off} \leq \kappa_s$, $P_{off} \leq \kappa_p$, for fixed constants $\kappa_s, \kappa_p$ set independently of the workflow under audit.
\end{definition}

\begin{definition}[Cost Non-Inflation]
    A pair of workflows satisfies cost non-inflation under a weight vector
    $(w_s, w_i, w_p)$ when $C_{off} \leq C_{on}$.
\end{definition}

Cost non-inflation is the property a regulator would wish to certify, since it states that leaving is no more effortful than joining. Its apparent difficulty is that it depends on effort weights that no study has established. The following observation removes that dependence.

\begin{proposition}[Weight Independence]
\label{prop:wi}
$C_{off} \leq C_{on}$ holds for every choice of non-negative weights $(w_s, w_i, w_p)$ if and only if $S_{off} \leq S_{on}$, $I_{off} \leq I_{on}$, and $P_{off} \leq P_{on}$.
\end{proposition}

\begin{proof}
For sufficiency, observe that
\[
C_{on} - C_{off} = w_s(S_{on}-S_{off}) + w_i(I_{on}-I_{off}) + w_p(P_{on}-P_{off}).
\]

Each parenthesized difference is non-negative by assumption and each weight is non-negative, so the sum is non-negative and $C_{off} \leq C_{on}$.

For necessity, suppose some component of the offboarding flow strictly exceeds its onboarding counterpart, say $S_{off} > S_{on}$. Setting $w_s = 1$ and $w_i = w_p = 0$ gives $C_{off} - C_{on} = S_{off} - S_{on} > 0$, so the constraint fails for that weight vector. The same construction applies to $I$ and $P$.\footnote{If weights are required to be strictly positive, necessity follows by taking $w_s = 1$ and $w_i = w_p = \varepsilon$ for sufficiently small $\varepsilon > 0$.} Componentwise domination is therefore necessary. \qed
\end{proof}

\begin{corollary}[Sufficiency of Structural Symmetry]
\label{cor:ss}
If a pair of workflows is structurally symmetric and $I_{off} \leq I_{on}$, then $C_{off} \leq C_{on}$ for every non-negative weight vector $(w_s, w_i, w_p)$.
\end{corollary}

\noindent Componentwise domination is immediate, and the claim follows from the sufficiency direction of Proposition~\ref{prop:wi}. The second condition is not an additional design burden. Cancellation ordinarily requires no data the provider has not already collected, so $I_{off} \leq I_{on}$ holds by construction in the financial setting, and the design problem reduces to capping steps and prompts.

The relative condition alone is insufficient against a strategic provider, since inflating $S_{on}$ and $P_{on}$ relaxes the bound on $S_{off}$ and $P_{off}$ without reducing exit friction in absolute terms. The absolute ceiling removes this degree of freedom. A provider cannot purchase additional exit friction by making enrollment more burdensome, because $\kappa_s$ and $\kappa_p$ are fixed independently of $S_{on}$ and $P_{on}$.

The argument is elementary, and that is its value. Verifying fair choice architecture does not require resolving how much cognitive effort a confirmation dialogue imposes relative to a text field. The weights may remain contested, and the compliance verdict does not move. Conversely, any workflow that inflates a single component admits some weighting under which exit is more costly than entry, which is the condition a manipulative design exploits.

We emphasize what this constraint does and does not establish. Satisfying structural symmetry and cost non-inflation does not certify that an interface is free of manipulation in general; it certifies the absence of one specific, structurally verifiable form of asymmetry, namely componentwise inflation of provider-controlled interaction quantities. A workflow can satisfy the constraint and still manipulate through channels the model does not count, including linguistic framing (Section~\ref{sec:invariants}), timing, or dark patterns outside the interaction-cost dimensions considered here. The contribution is a necessary, auditable floor, not a sufficient account of fair choice architecture. Similarly, the constraint does not by itself rule out a provider making both flows uniformly burdensome; it targets asymmetry between the two, not the absolute burden of either. Bounding absolute burden is a separate design goal, addressed in part by the ceilings $\kappa_s, \kappa_p$ introduced above.

\subsection{Worked Example}
Table~\ref{tab:cost_analysis} contrasts a representative enrollment flow with a cancellation flow of the kind documented in consumer-protection enforcement, including the FTC's complaint against Amazon alleging a multi-step cancellation flow internally named the ``Iliad Flow''~\cite{ftc2023amazon, gorman2023ftc_dark_patterns}; the counts are illustrative rather than measured from any particular interface. The scalar totals differ by only $1.0$, which understates the manipulation and is an artifact of one weighting among many. The structural reading is unambiguous. Exit inflates steps by three and prompts by two while requiring nothing new from the user, so by Proposition~\ref{prop:wi} there exist weightings under which exit is strictly more costly than entry.

The proposed model caps the offboarding sequence at three steps and one prompt. Structural symmetry holds and inputs fall to zero, so cost non-inflation follows for every non-negative weighting by Corollary~\ref{cor:ss}. Fairness does not require making exit as laborious as enrollment. It requires only that no provider-controlled dimension of exit exceed its counterpart at entry.

\begin{table*}[htbp]
\caption{Interaction cost analysis. Deltas are offboarding minus onboarding. A positive delta on a provider-controlled dimension ($S$ or $P$) indicates inflated exit friction. Illustrative weights $w_s=1.0$, $w_i=1.5$, $w_p=2.0$.}
\label{tab:cost_analysis}
\centering
\renewcommand{\arraystretch}{1.2}
\resizebox{0.8\linewidth}{!}{
\begin{tabular*}{\textwidth}{@{\extracolsep{\fill}}lccccl}
\toprule
\textbf{Workflow} & \textbf{Steps ($S$)} & \textbf{Inputs ($I$)} & \textbf{Prompts ($P$)} & \textbf{Cost ($C$)} & \textbf{Constraint} \\
\midrule
\multicolumn{6}{l}{\textit{Representative Asymmetric Model}} \\
\quad Sign-Up Flow & 3 & 4 & 1 & 11.0 & \\
\quad Cancellation Flow & 6 & 0 & 3 & 12.0 & \\
\quad Differential & $+3$ & $-4$ & $+2$ & $+1.0$ & Violated \\
\addlinespace
\multicolumn{6}{l}{\textit{Proposed Symmetric Model}} \\
\quad Sign-Up Flow & 3 & 4 & 1 & 11.0 & \\
\quad Cancellation Flow & 3 & 0 & 1 & 5.0 & \\
\quad Differential & $0$ & $-4$ & $0$ & $-6.0$ & Satisfied \\
\bottomrule
\end{tabular*}
}
\end{table*}

\subsection{Design Invariants}
\label{sec:invariants}
Structural symmetry constrains what a workflow counts, and it does not constrain how the workflow presents what it counts. Two interfaces with identical $S$, $I$, and $P$ may still differ in whether the user perceives exit as available. We therefore accompany the cost constraint with two invariants that bind the presentation layer, each stated so that it can be checked against an implemented component.

\textbf{Visual parity.} Primary actions that advance and complete a workflow must be rendered with equivalent styling, size, and visual prominence. Cancellation controls may not receive subordinate visual treatment, and no option may be preselected on the user's behalf. This invariant prevents visual friction from substituting for structural friction once step and prompt counts are bounded.

\textbf{Linguistic neutrality.} Instructional and confirmation text must be factual and non-directive. Cancellation interfaces often frame departure as a loss or cast doubt on the user's decision~\cite{mathur2019dark, gray2018dark}. Neutrality requires the interface to describe the consequences of each action without evaluating them. Together with structural symmetry, these invariants close the channels through which a bounded workflow could otherwise reintroduce the friction the cap removes.

\section{Prototype Implementation}
Following prior prototype-based evaluations of e-payment interfaces~\cite{kishnani2024towards, das2024design}, we instantiate the model as a high-fidelity mobile prototype of a credit card service, comprising parallel enrollment and cancellation workflows. The prototype simulates account states such as an outstanding balance and accrued reward points without live backend integration, since the properties under examination are properties of the interface.

\begin{figure}[htbp]
\centering
\includegraphics[width=0.7\linewidth]{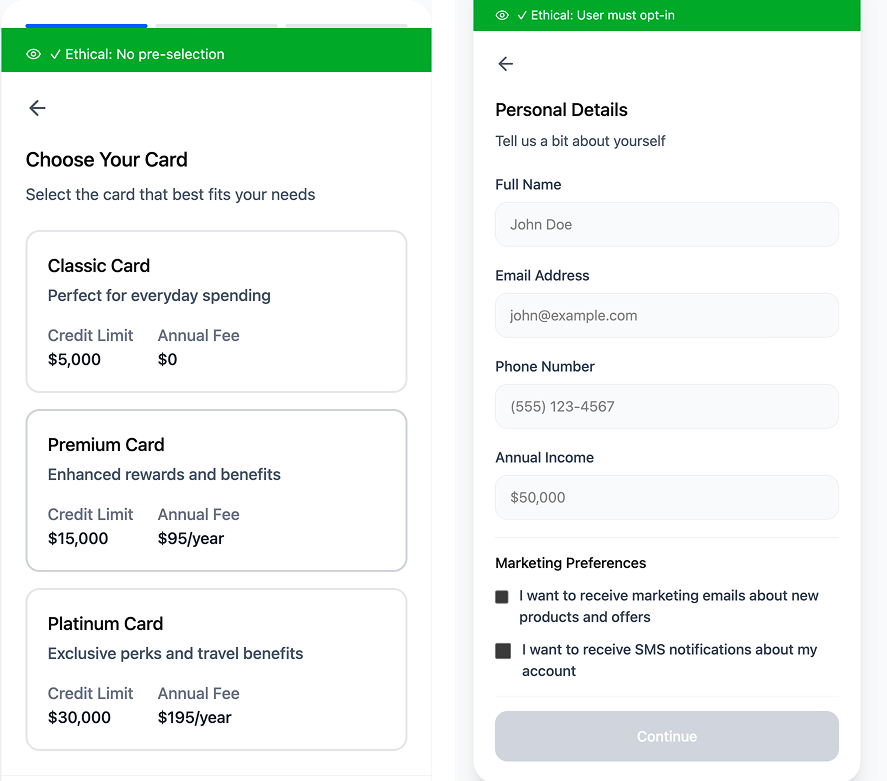}
\caption{Enrollment workflow in the prototype.}
\label{fig:onboarding}
\end{figure}

\begin{figure}[htbp]
\centering
\begin{subfigure}[b]{0.62\linewidth}
  \centering
  \includegraphics[width=\linewidth]{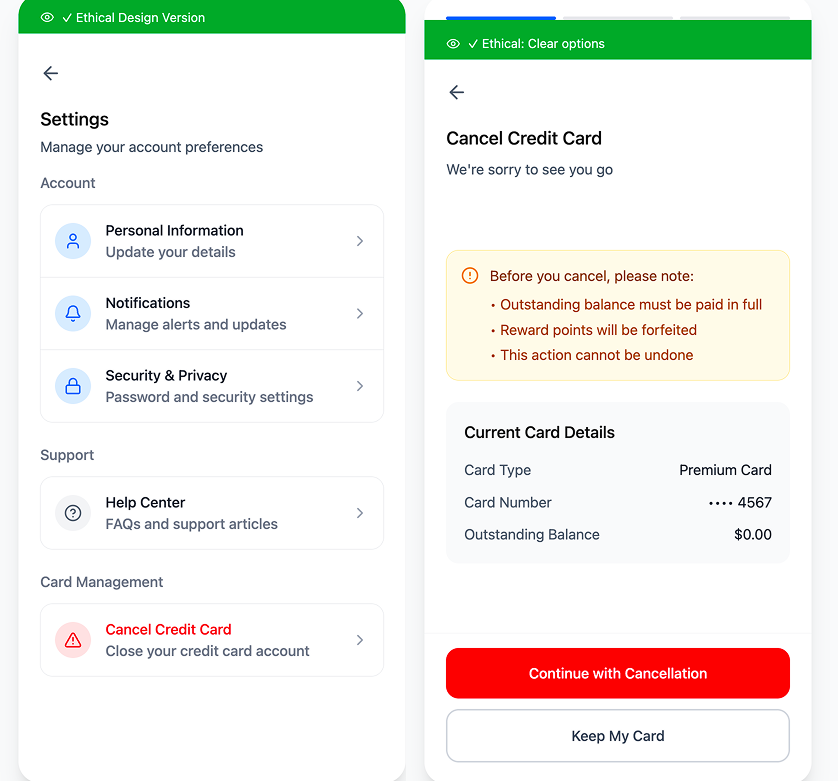}
  \caption{Cancellation workflow.}
  \label{fig:offboarding}
\end{subfigure}
\hfill
\begin{subfigure}[b]{0.30\linewidth}
  \centering
  \includegraphics[width=\linewidth]{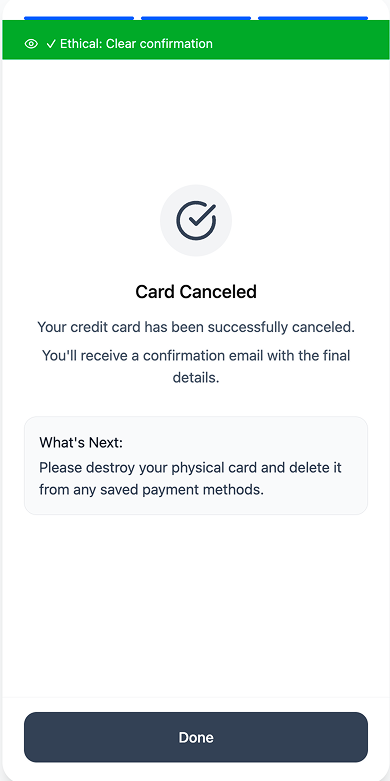}
  \caption{Confirmation screen.}
  \label{fig:confirmation}
\end{subfigure}
\caption{Cancellation workflow in the prototype and its confirmation screen.}
\label{fig:cancellation-flow}
\end{figure}

Both workflows traverse four screens connected by three mandatory transitions, realizing the symmetric row of Table~\ref{tab:cost_analysis}, where enrollment incurs $S_{on}=3$, $I_{on}=4$, $P_{on}=1$, and cancellation incurs $S_{off}=3$, $I_{off}=0$, $P_{off}=1$. Enrollment, counted from the authenticated home screen, comprises three transitions: to plan selection (step 1), to the personal and financial information screen where the four mandatory fields comprising $I_{on}$ are entered (step 2), and to registration confirmation, where a single prompt is presented (step 3). Cancellation, likewise counted from the authenticated home screen, comprises three transitions: to the account menu (step 1), to the cancellation request screen (step 2), and to account closure, where a single confirmation prompt is presented (step 3). The optional feedback screen is excluded from $S_{off}$ under the optional-screen exclusion defined in Section~\ref{sec:quantifying}.

Designing the two flows in parallel, rather than independently, enforces the structural budget, since each addition to one flow requires a corresponding justification in the other. The feedback step is optional by construction. Proposition~\ref{prop:wi} would permit a mandatory field here, since one input still falls below the four required at enrollment, but we adopt the stricter rule $I_{off} = 0$. Exit should demand nothing the provider does not already hold, and feedback is information it never needed.

Figure~\ref{fig:onboarding},~\ref{fig:offboarding} and~\ref{fig:confirmation} show the workflows. The design invariants of Section~\ref{sec:invariants} constrain the presentation layer. Primary actions such as continuation and cancellation are rendered with equivalent styling, spacing, and placement, and no option is preselected. Instructional and confirmation text is standardized to factual, non-directive phrasing, so that the interface states the consequence of each action without evaluating it.

We verified the prototype against the constraint by inspection. One author traversed both workflows and recorded $S$, $I$, and $P$ per screen, confirming that structural symmetry holds and that no primary action receives subordinate visual treatment or directive phrasing. This is a check on whether the artifact satisfies its specification, not evidence about how users behave.

\section{Discussion and Implications}
\textbf{Implications for Interface Design: }Identifying dark patterns has relied on subjective interpretation of when friction becomes excessive~\cite{mathur2019dark,mathur2021makes, lu2024awareness}. The proposed model resolves that friction into countable components, so a design team need not litigate whether a given retention prompt is manipulative. The prompt either fits within the exit budget or it does not. Ethical design thereby becomes a functional requirement rather than a guideline, and the counting procedure a designer applies during development is the same one available to an auditor afterward. Because the invariants attach to components rather than to individual screens, they could support automated validation pipelines that flag structural asymmetry and visual coercion in a component library before deployment, though we have not built such a checker.

\textbf{Regulatory and Policy Implications: }The \textit{as simple as} principle of the vacated Negative Option Rule~\cite{ftc2024negativeoption} outlived the rule itself, and what it has always lacked is an operational reading. Proposition~\ref{prop:wi} supplies one a regulator can apply without first commissioning a study to calibrate the relative burden of a dialogue against a text field, since the verdict does not depend on that calibration. A compliance test that requires only counting is also a test whose costs can be estimated, which is the ground on which the rule was struck down~\cite{customcomms2025}.

\textbf{Balancing Compliance and Autonomy: }Financial platforms operate under legal mandates requiring identity verification and record retention. These obligations bear on what a provider collects rather than on how many screens separate a user from an exit. Because the constraint caps steps and prompts while leaving inputs free to fall, a provider may satisfy every legal requirement of enrollment and still offer an exit that costs no more than entry. Where a jurisdiction genuinely mandates a step during cancellation, the constraint makes that step visible as a deliberate exception rather than allowing it to hide among discretionary friction.

\section{Limitations and Future Work}
This work contributes a formal constraint and a prototype that satisfies it, but provides no evidence of user behavior. The model and prototype are instantiated for a single domain, credit card enrollment and cancellation; other subscription-based services (streaming, SaaS, gym memberships) may involve workflow structures, such as multi-party contracts or regulatory verification steps, not captured by the three quantities considered here, and applying the constraint across domains and jurisdictions with differing compliance requirements is left to future work. Verification was conducted by inspection against the specification, leaving open whether symmetric flows improve completion time, error rate, or user trust. As stated in Section~\ref{sec:invariants}, the constraint is a necessary floor on one form of manipulation, not a validated model of user effort; a comparative study measuring completion time and user-reported fairness across symmetric and asymmetric variants of the prototype is the natural next step. The effort weights are illustrative, although the compliance verdict does not depend on them. 

Future behavioral studies should test the model's underlying assumption that the three counted quantities capture the effort users actually experience, and whether visual parity and linguistic neutrality remain perceptually distinct from structural symmetry as they are in the specification. Deploying the constraint in live services also raises the question of whether the structural budget survives jurisdiction-specific retention and verification mandates, which we treat as affecting inputs alone.

\section{Conclusion}
Asymmetric cancellation friction is manipulation carried out against a user who has already decided to leave. We modeled the provider as an adversary whose instrument is effort and expressed fairness as a single constraint over the interaction costs it controls. Offboarding need not cost the same as onboarding. It must only never cost more. Because the comparison depends on counting steps, inputs, and prompts rather than on how heavily each is weighted, the criterion can be evaluated directly on an implemented interface, and disagreement about the cost of any single obstacle need not be resolved before reaching a verdict. Designers gain a budget rather than a debate, and auditors and regulators gain a test they can apply.

\section*{Acknowledgements}
We acknowledge the Data Agency and Security (DAS) Lab at George Mason University (GMU), and Google for partially supporting this work. The opinions expressed are solely those of the authors.

\bibliographystyle{splncs04}
\bibliography{references}

\end{document}